\documentclass{article} %
\usepackage{nips14submit_e,times}
\usepackage{hyperref}

\usepackage{graphicx}
\usepackage{amsmath}
\usepackage{amsfonts}
\usepackage{amssymb}
\usepackage{amsthm}
\usepackage[sectionbib]{natbib}
\usepackage{algorithm}
\usepackage[noend]{algorithmic}
\usepackage{subfigure}
\usepackage{url}
\usepackage{float}

\usepackage{authblk}

\restylefloat{table}

\newcommand{\defemph}[1]{\textbf{#1}}

\newcommand{\R}{\mathbb{R}}
\newcommand{\N}{\mathbb{N}}

\newcommand{\bigO}{\mathcal{O}}

\newcommand{\bigOmega}{\Omega}

\newcommand{\cD}{\mathcal{D}}
\newtheorem{pr}{Proposition}
\newtheorem{thm}[pr]{Theorem}

\newtheorem{df}[pr]{Definition}

\newcommand{\loss}{f}
\newcommand{\cummloss}[2]{L_{#1}(#2)}
\newcommand{\lossbound}[2]{\textbf{L}_{#1}(#2)}

\newcommand{\comm}[2]{C_{#1}(#2)}
\newcommand{\commbound}[2]{\textbf{C}_{#1}(#2)}

\newcommand{\totaltime}{T}
\newcommand{\inputSpace}{Z}
\newcommand{\modelClass}{W}
\newcommand{\modelconf}{\mathbf{w}}

\newcommand{\uprule}{\varphi}
\newcommand{\onlineAlgo}{A}
\newcommand{\distProtocol}{P}
\newcommand{\syncop}{\sigma}
\newcommand{\divergence}{\delta}
\newcommand{\violations}{V}
\newcommand{\commcost}[1]{c_{#1}}
\newcommand{\commPerViolation}{\textbf{c}_k}
\newcommand{\radius}[1]{D_{#1}}

\newcommand{\projection}[1]{\pi_{\Gamma}\left(#1\right)}

\title{Adaptive Communication Bounds for\\Distributed Online Learning}
\date{}

\author[1]{Michael Kamp}
\author[1]{Mario Boley}
\author[2]{Michael Mock}
\author[3]{Daniel Keren}
\author[4]{Assaf Schuster}
\author[4]{Izchak Sharfman}
\affil[1]{Fraunhofer IAIS \& University Bonn\\\texttt{\{surname.name\}@iais.fhg.de}}
\affil[2]{Fraunhofer IAIS\\\texttt{michael.mock@iais.fhg.de}}
\affil[3]{Haifa University\\\texttt{dkeren@cs.haifa.ac.il}}
\affil[4]{Technion, Israel Institute of Technology\\\texttt{\{assaf,tsachis\}@technion.ac.il}}

\nipsfinalcopy 

\begin{document}

\maketitle

\begin{abstract}
We consider distributed online learning protocols that control the exchange of information between local learners in a round-based learning scenario. The learning performance of such a protocol is intuitively optimal if approximately the same loss is incurred as in a hypothetical serial setting. If a protocol accomplishes this, it is inherently impossible to achieve a strong communication bound at the same time. In the worst case, every input is essential for the learning performance, even for the serial setting, and thus needs to be exchanged between the local learners. However, it is reasonable to demand a bound that scales well with the hardness of the serialized prediction problem, as measured by the loss received by a serial online learning algorithm. We provide formal criteria based on this intuition and show that they hold for a simplified version of a previously published protocol.
\end{abstract}

\section{Introduction}
We consider round-based learning scenarios on multiple connected dynamic data streams where a real-time service is provided by a \defemph{distributed online learning system} of $k\in\N$ local learners.
A \defemph{distributed online learning protocol} controls the exchange of information between the learners with the goal of providing a service quality---measured by a loss function---that is close to the one of a hypothetical serial learner. Such an optimal predictive behavior can be trivially achieved by centralizing all data but at the expense of a total communication cost in $\bigOmega(k\totaltime)$ for a total time horizon $\totaltime\in\N$. This amount of communication in practice often exceeds network capacities or delays the service beyond its time limits. No communication, on the other hand, leads to a significantly higher loss that increases with the number of learners.
Earlier research focused on static communication strategies that retain the service quality of a hypothetical centralized learner by communicating periodically, thereby reducing communication by a fixed factor, but not changing the asymptotic behavior. 
In this paper, we introduce the notion of an adaptive protocol that invests communication only if it thereby significantly reduces its loss. As a consequence, the communication bound of an adaptive protocol depends on the loss of the serial setting instead of the total time. 
We show for a simplified version of dynamic averaging, a previously published protocol, that it satisfies this notion of adaptivity while it retains the optimal predictive behavior of a serial learner.

\section{Performance Bounds for Distributed Online Learning Protocols}
In the following, we assume an online learning algorithm $\onlineAlgo=(\modelClass,\uprule, \loss)$ run on each local learner $l\in[k]$ in the distributed system to maintain a local model $w_{t,l}\in\modelClass$. At each time point $t\in\N$, each learner observes an input $z_{t,l}$ drawn independently from an input space $\inputSpace$ from a time-variant distribution $\cD_t: \inputSpace\rightarrow [0,1]$. Based on this input and the local model, the local learner provides a service whose quality is measured by a loss function $\loss:\inputSpace\times \modelClass\rightarrow \R_+$. After providing the service, the local learner updates its local model using an update rule $\uprule: \inputSpace\times \modelClass\rightarrow \modelClass$. The performance of $A$ within a time horizon $\totaltime\in\N$ is measured by its \defemph{cumulative loss}
\[
\cummloss{\onlineAlgo}{\totaltime}=\sum_{t=1}^\totaltime f(z_t,w_t)\enspace .
\]
Performance guarantees for online learning algorithms are typically given by a \defemph{loss bound} $\lossbound{\onlineAlgo}{\totaltime}$, i.e., for all input sequences $z_1,...,z_T\in\inputSpace$ it holds that $\cummloss{\onlineAlgo}{\totaltime}\leq \lossbound{\onlineAlgo}{\totaltime}$. Loss bounds can be defined with respect to a sequence of reference models, in which case they are referred to as regret bounds. Such regret bounds are typically sub-linear in $\totaltime$, e.g., for scenarios with a static distribution an optimal regret bound is in $\bigO(\sqrt{\totaltime})$ (see, e.g.,~\citet{cesa-bianchi/book/2006}), which is achieved by many online learning algorithms, including stochastic gradient descent~~\citep{zinkevich/nips/2010} and passive aggressive updates~\citep{crammer/jmlr/2006}.
A distributed online learning protocol $\distProtocol=(\onlineAlgo,\syncop)$ runs algorithm $A$ on a distributed online learning system and interchanges information between the local learners by synchronizing their local models $w_{t,1},...,w_{t,k}$ using a \defemph{synchronization operator} $\syncop:\modelClass^k\rightarrow\modelClass^k$.
The cumulative loss of a distributed online learning protocol $P$ is defined as
\[
\cummloss{\distProtocol}{\totaltime,k}\sum_{t=1}^\totaltime\sum_{l=1}^k f(z_{t,l},w_{t,l})\enspace .
\]
We measure the amount of communication of $\syncop$ by its \defemph{communication cost} $\commcost{\syncop}:\modelClass^k\times\N\rightarrow\N$, i.e., $\commcost{\syncop}(\modelconf_t,t)$ is the number of bytes transmitted at time point $t$ when $\syncop$ is applied to the current model configuration $\modelconf_t=w_{t,1},...,w_{t,k}$. The communication cost depends on the system architecture, e.g., in a system with a designated coordinator node all models can be sent to the coordinator, processed and the new models can be send back to the learners using $2k$ messages each containing one model of fixed size, thus $\commcost{\syncop}(\modelconf_t,t)\in\Theta(k)$.
We denote the cumulative amount of communication required for synchronization by 
\[
\comm{\distProtocol}{\totaltime,k}=\sum_{t=1}^\totaltime \commcost{\syncop}(\modelconf_t,t)\enspace . 
\]
There is a natural trade-off between communication and loss of a distributed online learning system. A loss similar to a serial setting can be trivially achieved by permanent centralization. On the other hand, communication can be entirely omitted. If the cumulative loss of an online learning algorithm $\onlineAlgo$ is bounded by $\lossbound{\onlineAlgo}{\totaltime}$, the loss of a distributed system with $k$ local learners running $\onlineAlgo$ without any synchronization is bounded by $\lossbound{\text{nosync}}{\totaltime,k} = k \lossbound{\onlineAlgo}{\totaltime}$, whereas the communication is $\comm{\distProtocol}{\totaltime}=0$. Such a distributed system processes $k\totaltime$ inputs. The loss of a permanently centralizing system on the other hand is bounded by $\lossbound{\text{central}}{\totaltime, k} = \lossbound{\onlineAlgo}{k\totaltime}$, i.e., the loss bound of a serial online learning algorithm processing $k\totaltime$ inputs. 
For example in a static scenario with an optimal loss bound of $\lossbound{\onlineAlgo}{k\totaltime}\in\bigO(\sqrt{k\totaltime})$, the loss bound of centralization is superior by a factor of $\sqrt{k}$.
However, the communication cost $\comm{\distProtocol}{\totaltime}$ is in $\Theta(k\totaltime)$. Current communication strategies that retain the service quality of a hypothetical centralized learner are based on communicating periodically~\citep{mann/nips/2009,dekel/jmlr/2012} after a fixed number $b\in\N$  of data points have been processed, thereby reducing communication by a fixed factor of $1/b$, but the communication bound remains in $\Theta(k\totaltime)$. The communication bound of an adaptive protocol should only depend on $\lossbound{\onlineAlgo}{\totaltime}$ and not on $\totaltime$, while at the same time retaining the loss bound of the serial setting. In the following definition we formalize this.
\begin{df}
A distributed online learning protocol $\distProtocol=(\onlineAlgo,\syncop)$ processing $k\totaltime$ inputs is \defemph{consistent} if it retains the loss bound of the serial online learning algorithm $\onlineAlgo$ processing $kT$ inputs, i.e., 
\[
\lossbound{\distProtocol}{\totaltime,k}\in \bigO\left( \lossbound{\onlineAlgo}{k\totaltime}\right)\enspace .
\]
The protocol is \defemph{adaptive} if its communication bound is linear in the number of local learners $k$ and the loss bound $\lossbound{\onlineAlgo}{k\totaltime,k}$ of the serial online learning algorithm, i.e.,
\[
\commbound{\distProtocol}{\totaltime,k}\in \bigO\left( k \lossbound{\onlineAlgo}{k\totaltime} \right)\enspace .
\]
\end{df}
An efficient protocol is adaptive and consistent at the same time. In the following we will present such a protocol.

\section{An Adaptive and Consistent Distributed Online Learning Protocol}
In~\citet{kamp2014communication} we presented a dynamic protocol, denoted \defemph{dynamic averaging}, that adapts its communication to the loss incurred. For specific scenarios this protocol is consistent. In particular, the protocol is consistent for an online learning scenario where each learner maintains a linear model, i.e., $\modelClass=\R^n$, using a specific type of update rules denoted $f$-proportional convex update rules for a loss function $\loss$. That is, there exists a constant $\gamma > 0$, a closed convex set $\Gamma_{z} \subseteq \R^n$, and $\tau_{z} \in (0,1]$ such that for all $w \in \R^n$ and $z \in \inputSpace$ it holds that
\begin{enumerate}
\item[(i)] $\|w-\uprule(z,w)\|\geq\gamma \loss(z,w)$, i.e., the update magnitude is a true fraction of the loss incurred, and 
\item[(ii)] $\uprule(z,w)=w+\tau_{z}\left( \projection{w}-w \right)$ where $\projection{w}$ denotes the projection of $w$ onto $\Gamma_{z}$, i.e., the update direction is identical to the direction of a convex projection that only depends on the training example.
\end{enumerate}
Examples of such update rules are stochastic gradient descend, as well as passive aggressive updates and their regularized variants. For these update rules the update magnitude is bounded by a multiple of the loss, i.e., there exists a constant $C\in\R_+$ such that for all inputs $z\in\inputSpace$ it holds that $\|w-\uprule(z,w)\|\leq C f(z,w)$. This is true, e.g., for stochastic gradient descent with $C=\radius{\inputSpace}\radius{\modelClass}$, where $D$ denotes the diameter. 
In this setting dynamic averaging retains the regret bounds of static averaging, i.e., synchronizing every $b\in\N$ rounds. 
In case of a fixed target distribution $\cD$ it has been shown by~\citet{dekel/jmlr/2012} that static averaging retains the optimality of stochastic gradient descent, i.e., it retains the regret bound of the serial online learning algorithm. Thus, both static and dynamic averaging are consistent. However, static averaging has communication costs of $\comm{\text{static}}{T,k}=\commPerViolation\lceil T/b \rceil$ and is thus not adaptive. In the following we will define dynamic averaging and provide a communication bound in $\lossbound{\onlineAlgo}{k\totaltime}$.

The dynamic averaging protocol $\distProtocol=(\onlineAlgo,\syncop_{\Delta})$ synchronizes the local learners using a \defemph{dynamic averaging operator} $\syncop_{\Delta}$. This operator only communicates when the \defemph{model divergence} ${\divergence(\modelconf_t)=1/k \sum_{l=1}^k\|w_{t,l}-\overline{\modelconf}_t\|^2}$ exceeds a divergence threshold $\Delta$, where $\overline{\modelconf}_t=1/k \sum_{l=1}^k w_{t,l}$ denotes the average model. 
The dynamic averaging operator is defined as
\[
\syncop_\Delta(\modelconf_t)=\begin{cases}(\overline{\modelconf}_t,\dots,\overline{\modelconf}_t), &\text{ if } \divergence(\modelconf_t)>\Delta\\
\modelconf_t, &\text{ otherwise}\\
\end{cases}\enspace .
\]
In order to decide when to communicate, each local learner $l\in [k]$  monitors the \defemph{local condition} $\|w_{t,l}-r_t\|^2\leq\Delta$ for a \defemph{reference vector} $r_t\in\modelClass$ that is common among all learners (see~\citet{keren2012shape,sharfman/tods/2007,gabel/IPDPS/2014,giatrakos2012prediction} for a more general description of this method). The local conditions guarantee that if none of them is violated the divergence does not exceed the threshold $\Delta$. 
If a local condition is violated, a synchronization is triggered. We refer to the averaging of all models as \defemph{full synchronization}. 
Using an architecture with $k$ local learners and an additional coordinator node, a full synchronization can be performed with communication linear in $k$, i.e., $\commcost{\syncop_{\Delta}}(\modelconf_t,t)=\commPerViolation\in\bigO(k)$. In order to further reduce communication, instead of a full synchronization, the distributed system can try to resolve the violation by averaging a small subset of local models so that the local conditions hold after this averaging. We refer to this operation as \defemph{partial synchronization}. 
The subset is augmented by local models until either all local conditions hold or a stop criterion of an appropriate hedging strategy is met. The hedging strategy ensures that the communication in each time step cannot exceed $\commPerViolation$.
In the following we show that dynamic averaging is adaptive.
\begin{thm}
The communication $\comm{\distProtocol}{\totaltime,k}$ of a distributed online learning protocol using the dynamic averaging operator and an $f$-proportional convex update rule with ${\|w-\uprule(z,w)\|\leq C f(z,w)}$ is bounded by
\[
\commbound{\distProtocol}{\totaltime,k} = \commPerViolation\frac{C}{\sqrt{\Delta}}\lossbound{\distProtocol}{\totaltime,k}
\]
where $\commPerViolation$ is an upper bound on the amount of communication per time point $t$.
\end{thm}
\begin{proof}
The dynamic averaging protocol communicates only if a violation of a local condition  $\|w_{t,l}-r_t\|^2\leq\Delta$ occurs. At each time point with at least one violation dynamic averaging has communication costs of at most $\commPerViolation$, i.e., the cost of a full synchronization. Thus, we can bound the amount of communication by bounding the number of violations. That is, we derive a bound for $\violations_l(\totaltime)$, the number of time points $t\in[\totaltime]$ where the local condition of learner $l$ is violated. 
For that, assume that at $t=1$ all models are initialized with $w_{1,1}=\dots=w_{1,k}$ and $r_1=\overline{W}_1$, i.e., for all local conditions it holds that $\|w_{1,l}-r_1\|=0$. A violation, i.e., $\|w_{t,l}-r\|>\sqrt{\Delta}$, occurs if one local learner drifts away from $r_t$ by more than $\sqrt{\Delta}$. After a violation a full synchronization is performed and $r_t=\overline{W}_t$, hence $\|w_{t,l}-r_t\|=0$ and the situation is again similar to the initial setup for $t=1$. In the worst case, a local learner drifts continuously in one direction until a violation occurs. Hence, we can bound the number of violations $\violations_l(\totaltime)$ by the sum of its drifts divided by $\sqrt{\Delta}$:
\[
\violations_l(\totaltime)\leq \frac{1}{\sqrt{\Delta}}\sum_{t=1}^\totaltime\|w_{t,l}-w_{t+1,l}\|=\frac{1}{\sqrt{\Delta}}\sum_{t=1}^\totaltime\|w_{t,l}-\uprule(z_{t,l},w_{t,l})\|\leq\frac{1}{\sqrt{\Delta}}\sum_{t=1}^\totaltime C\loss(z_{t,l},w_{t,l})\enspace .
\]
To bound the communication we need to bound the number of time points $t\in[\totaltime]$ where at least one learner $l$ has a violation, denoted $\violations(\totaltime)$. In the worst case, all violations at all local learners occur at different time points,so that we can upper bound $\violations(\totaltime)$ by the sum of local violations $\violations_l(\totaltime)$ which is again upper bounded by the cumulative sum of drifts of all local models:
\[
\violations(\totaltime)\leq \sum_{l=1}^k\violations_l(\totaltime) \leq \frac{1}{\sqrt{\Delta}}\sum_{t=1}^\totaltime\sum_{l=1}^k C\loss(z_{t,l},w_{t,l})=\frac{C}{\sqrt{\Delta}}\cummloss{\distProtocol}{\totaltime,k}\enspace .
\]
Since dynamic averaging has communication costs of at most $\commPerViolation$ per time point, the total amount of communication is 
\[
\comm{\distProtocol}{\totaltime,k} = \commPerViolation \violations(\totaltime)\leq \commPerViolation\frac{C}{\sqrt{\Delta}}\cummloss{\distProtocol}{\totaltime,k}\leq\underbrace{\commPerViolation}_{\in\bigO(k)}\frac{C}{\sqrt{\Delta}}\lossbound{\distProtocol}{\totaltime,k}\in\bigO\left(k\lossbound{\distProtocol}{\totaltime,k}\right)
\]
\end{proof}
Since dynamic averaging is consistent in this setting, $\lossbound{\distProtocol}{\totaltime,k}\in\bigO(\lossbound{\onlineAlgo}{k\totaltime})$ and we can follow that $\commbound{\distProtocol}{\totaltime,k}\in\bigO\left(k\lossbound{\onlineAlgo}{k\totaltime}\right)$, i.e., dynamic averaging is adaptive.

\section{Conclusion}
The analysis of the dynamic averaging protocol constitutes a first example of a relatively specific setting in which adaptivity and consistency can be achieved at the same time. Central for the future is to adapt dynamic averaging so that this holds also for more general settings. Moreover, tighter communication bounds are desirable which could be achieved by taking into account the communication reduction potential of partial synchronization and hedging strategies.

\subsubsection*{Acknowledgments}
This research has been partially supported by the EU FP7-ICT-2013-11 under grant 619491 (FERARI).

\begin{small}
\bibliographystyle{plainnat}
\bibliography{biblio}
\end{small}
\end{document}